\documentclass[a4paper,UKenglish,autoref,numberwithinsect]{mylipics2019}

\title{The Iteration Number of Colour Refinement}

\author{Sandra Kiefer}{RWTH Aachen University, Aachen, Germany}{kiefer@cs.rwth-aachen.de}{}{}

\author{Brendan D.\ McKay}{Australian National University, Canberra, Australia}{brendan.mckay@anu.edu.au}{}{}

\authorrunning{S.\ Kiefer and B.\ McKay}

\nolinenumbers

\usepackage{mathtools}
\usepackage{anyfontsize}
\usepackage{tikz}
\usetikzlibrary{arrows}
\usetikzlibrary{calc}

\newcounter{claimcounterglobal}
\newcounter{claimcounter}
\renewenvironment{claim}[1][]{
  
  \medbreak\par\noindent%
  \ifthenelse{\equal{#1}{}}{%
      \setcounter{claimcounter}{0}%
      \refstepcounter{claimcounterglobal}%
      \refstepcounter{claimcounter}%
      \textit{Claim~\arabic{claimcounter}.}
  }{%
    \ifthenelse{\equal{#1}{resume}}{%
      \refstepcounter{claimcounterglobal}%
      \refstepcounter{claimcounter}%
      \textit{Claim~\arabic{claimcounter}.}
    }{%
      \textit{Claim~#1.}
    }
  }
}{
  \par\medbreak
}

\newcommand{\llcurly}{\{\hspace{-3.5pt}\{}
\newcommand{\rrcurly}{\}\hspace{-3.5pt}\}}

\newcommand{\uenda}{\tag*{$\lrcorner$}} 

\DeclareMathOperator{\WL}{WL}

\theoremstyle{plain}
\newtheorem{notation}     [theorem] {Notation}

\makeatletter
\def\hlinewd#1{%
  \noalign{\ifnum0=`}\fi\hrule \@height #1 \futurelet
   \reserved@a\@xhline}
\makeatother

\begin{document}

\maketitle

\begin{abstract}
 The Colour Refinement procedure and its generalisation to higher dimensions, the Weisfeiler-Leman algorithm, are central subroutines in approaches to the graph isomorphism problem. In an iterative fashion, Colour Refinement computes a colouring of the vertices of its input graph. 

 A trivial upper bound on the iteration number of Colour Refinement on graphs of order $n$ is $n-1$. We show that this bound is tight. More precisely, we prove via explicit constructions that there are infinitely many graphs $G$ on which Colour Refinement takes $|G|-1$ iterations to stabilise. Modifying the infinite families that we present, we show that for every natural number $n \geq 10$, there are graphs on~$n$ vertices on which Colour Refinement requires at least $n-2$ iterations to reach stabilisation.
\end{abstract}

\section{Introduction}

Colour Refinement, which is also known as Naïve Vertex Classification or the 1-dimensional Weisfeiler-Leman algorithm (1-WL), is an important combinatorial algorithm in theoretical and practical approaches to the graph isomorphism problem. In an iterative fashion, it refines an isomorphism-invariant partition of the vertex set of the input graph. This process stabilises at some point and the final partition can often be used to distinguish non-isomorphic graphs \cite{BabErdSelSta80}. Colour Refinement can be implemented to run in time $O((m+n) \log n)$, where $n$ is the order of the input graph and $m$ is its number of edges \cite{carcro82,mck81}. Most notably, its efficient implementations are used in all competitive graph isomorphism solvers (such as~\texttt{Nauty} and~\texttt{Traces}~\cite{mckaypip14}, \texttt{Bliss}~\cite{JunttilaK07} and \texttt{saucy}~\cite{DargaLSM04}). 

Colour Refinement has been rediscovered many times, one of its first occurences being in a paper on chemical information systems from the 1960s \cite{mor65}. The procedure is applied in plenty of other fields, for example, it can be modified to reduce the dimension of linear programs significantly~\cite{GroheKMS14}. Other applications are in the context of graph kernels~\cite{ShervashidzeSLMB11} or static program analysis~\cite{LiSSSS16}. A recently discovered connection to deep learning shows that the expressive power of Colour Refinement is captured by graph neural networks \cite{morritfey+19}.

As described above, Colour Refinement computes a stable colouring of its input graph. It is known that two given graphs result in equal colourings, i.e.\ are not distinguished by Colour Refinement, if and only if there is a fractional isomorphism between them \cite{god97,ramscheiull94,tin91}. Moreover, the graphs which Colour Refinement identifies up to isomorphism (i.e.\ distinguishes from all non-isomorphic ones) have been completely characterised \cite{DBLP:journals/cc/ArvindKRV17,kieschweisel15}. 

To obtain its final colouring, the algorithm proceeds in iterations. In this paper, we investigate how many iterations it takes for the algorithm to terminate. More specifically, for~$n \in \mathbb{N}$, we are interested in~$\WL_1(n)$, the maximum number of iterations required to reach stabilisation of Colour Refinement among all graphs of order~$n$.

While not directly linked to the running time on a sequential machine, the iteration number corresponds to the parallel running time of Colour Refinement (on a standard PRAM model) \cite{groverb06,KoblerV08}. Furthermore, via a connection to counting logics, a bound on the iteration number for graphs of a fixed size directly translates into a bound on the descriptive complexity of the difference between the two graphs, namely into a bound on the quantifier depth of a distinguishing formula in the extension of the 2-variable fragment of first-order logic by counting quantifiers \cite{caifurimm92,immlan90}. Moreover, the iteration number of 1-WL equals the depth of a graph neural network that outputs the stable vertex colouring of the underlying graph with respect to Colour Refinement \cite{morritfey+19}. 

Considering paths, one quickly determines that~$\WL_1(n)\geq \frac{n}{2} -1$ holds for every $n \in \mathbb{N}$. By contrast, on random graphs, the iteration number is asymptotically almost surely~$2$~\cite{BabErdSelSta80}. The best published lower bound on the iteration number of Colour Refinement on $n$-vertex graphs is $n-O(\sqrt{n})$~\cite{krever15}. Concerning the upper bound, the trivial inequality~$\WL_1(n) \leq n-1$ holds for every repeated partitioning of a set of size~$n$ and it does not take into account any further properties of the input graph or of the algorithm used to execute the partitioning. Still, no improvement over this upper bound has been established.

Our first main result reads as follows.

\begin{theorem}\label{thm:main:infinite}
 For every $n \in \mathbb{N}_{\geq 10}$ with $n = 12$ or $n \bmod 18 \notin \{6,12\}$, it holds that $\WL_1(n) = n-1$. 
\end{theorem}

Thus, there are infinitely many $n \in \mathbb{N}$ with $\WL_1(n) = n-1$. We can even determine the iteration number up to an additive constant of 1 for all $n \in \mathbb{N}$ (where the precise numbers for $n \leq 9$ can easily be determined computationally), as stated in our second main result.

\begin{theorem}\label{thm:main:n-2}
 For every $n \in \mathbb{N}_{\geq 10}$, it holds that $\WL_1(n) \in \{n-2, n-1\}$. 
\end{theorem}

We obtain our bounds via an empirical approach. More precisely, we have designed a procedure that enables us to systematically generate for all $n \leq 64$ graphs of order $n$ that obey certain constraints (to render the procedure tractable) and on which Colour Refinement takes $n-1$ iterations to stabilise. Analysing the graphs, we determined the connections between colour classes during the execution of the algorithm in detail. If the vertex degrees that are present in the graph are low, then the connections between colour classes of size 2 are restricted. This allows us to develop an elegant graphical visualisation and a compact string representation of the graphs with low vertex degrees that take $n-1$ iterations to stabilise. Using these encodings, we are able to provide infinite families with $n-1$ Colour Refinement iterations until stabilisation.  

Our analysis enables a deep understanding of the families that we present. Via slight modifications of the graph families, we can then cover a large portion of graph sizes and, allowing to go from connected graphs to general graphs, we can construct the graphs that yield Theorem \ref{thm:main:n-2}.

\paragraph*{Related work}

Colour Refinement is the 1-dimensional version of the so-called Weisfeiler-Leman algorithm. For every $k \in \mathbb{N}$, there exists a generalisation of it ($k$-WL), which colours vertex $k$-tuples in the input graph instead of single vertices only. See \cite{kiefer} for an in-depth study of the main parameters of Colour Refinement and $k$-WL. 

Similarly as for Colour Refinement, one can consider the number $\WL_k(n)$ of iterations of $k$-WL on graphs of order $n$. Notably, contrasting our results for Colour Refinement, in \cite{kieschwe16}, it was first proved that the trivial upper bound of $\WL_2(n) \leq n^2-1$ is not even asymptotically tight (see also the journal version \cite{kieschw19}). This foundation fostered further work, leading to an astonishingly good new upper bound of $O(n \log n)$ for the iteration number of 2-WL \cite{lichponschwei19}.

For fixed~$k>1$, it is already non-trivial to show linear lower bounds on~$\WL_k(n)$. Modifying a construction of Cai, F\"urer, and Immerman~\cite{caifurimm92}, this was achieved by F\"{u}rer \cite{Furer01}, who showed that~$\WL_k(n) \in \Omega(n)$, remaining to date the best known lower bound when the input is a graph. Only when considering structures with relations of higher arity than 2 as input, better lower bounds on the iteration number of~$k$-WL have been proved~\cite{BerkholzN16}. 

For $k > 2$, regarding upper bounds on the iteration number of $k$-WL, without further knowledge about the input graph, no significant improvements over the trivial upper bound $n^k-1$ are known.\footnote{Note that the bound $n^k-1$ is not tight, since the initial partition of the $k$-tuples already has multiple classes, for example, one consisting of all tuples of the form $(v,v,\dots,v)$.} Still, when the input graph has bounded treewidth or is a 3-connected planar graph, polylogarithmic upper bounds on the iteration number of $k$-WL needed to identify the graph are known \cite{groverb06,verb07}.

Although for every natural number $k$, there are non-isomorphic graphs that are not distinguished by~$k$-WL~\cite{caifurimm92}, it is known that for every graph class with a forbidden minor, a sufficiently high-dimensional Weisfeiler-Leman algorithm correctly decides isomorphism~\cite{Grohe12}. Recent results give new upper bounds on the dimension needed for certain interesting graph classes \cite{grokie19,groneu19}. A closely-related direction of research investigates what properties the Weisfeiler-Leman algorithm can detect in graphs \cite{ArvindFKV2018,fuhlkoebverb20,fur17}.

\section{Preliminaries}

By $\mathbb{N}$, we denote the set of natural numbers, i.e.\ $\{1,2,\dots\}$. We set $\mathbb{N}_0 \coloneqq \mathbb{N} \cup \{0\}$ and, for $k,\ell \in \mathbb{N}_0$, we define $[k,\ell] \coloneqq \{n \in \mathbb{N}_0 \mid k \leq n \leq \ell\}$ and $[k] \coloneqq [1,k]$. For a set $S$, a \emph{partition} of $S$ is a set $\Pi$ of non-empty sets such that $\bigcup_{M \in \Pi} M = S$ and for all $M, M' \in \Pi$ with $M \neq M'$, it holds that $M \cap M' = \emptyset$. For two partitions~$\Pi$ and~$\Pi'$ of the same set $S$, we say that~$\Pi'$ is \emph{finer} than~$\Pi$ (or $\Pi'$ \emph{refines}~$\Pi$) if every element of~$\Pi'$ is a (not necessarily proper) subset of an element of~$\Pi$. We write~$\Pi \succeq \Pi'$ (and equivalently~$\Pi' \preceq \Pi$) to express that~$\Pi'$ is finer than~$\Pi$. Concurrently, we say that~$\Pi$ is \emph{coarser} than~$\Pi'$. If both $\Pi \succeq \Pi'$ and $\Pi' \succeq \Pi$ hold, we denote this by $\Pi \equiv \Pi'$.

For $S \neq \emptyset$, the partition $\{S\}$ is the \emph{unit} partition of $S$. The partition $\big\{ \{s\} \mid s \in S\big\}$ is called the \emph{discrete partition} of $S$. A set of cardinality 1 is a \emph{singleton}.

All graphs that we consider in this paper are finite and simple, i.e.\ undirected without self-loops at vertices. For a graph with vertex set $V(G)$ and edge set $E(G)$, its \emph{order} is $|G| \coloneqq |V(G)|$. For a vertex $v \in V(G)$, we denote by $N(v)$ the neighbourhood of $v$ in $G$, i.e.\ the set $\{w \mid \{v,w\} \in E(G)\}$. Similarly, for a vertex set $W$, we set $N(W) \coloneqq \big\{v \mid v \notin W, \exists w\in W \colon v \in N(w)\big\}$. The \emph{degree} of a vertex $v$ is $\deg(v) \coloneqq |N(v)|$ (since the graph $G$ will be clear from the context, we do not need to include it in our notation). We also set $\deg(G) \coloneqq \{\deg(v) \mid v \in V(G)\}$.
If there is a $d \in \mathbb{N}_0$ such that $\deg(G) = \{d\}$, the graph $G$ is $d$-regular. A \emph{regular} graph is a graph that is $d$-regular for some $d \in \mathbb{N}_0$. By a \emph{matching}, we mean a 1-regular graph. 

Let $G$ be a graph with at least two vertices. If there are sets $\emptyset \neq P, Q \subseteq V(G)$ such that $V(G) = P \cup Q$ and $P \cap Q = \emptyset$ and $E(G) \cap \{\{v,w\} \mid v,w \in P\} = \emptyset = E(G) \cap \{\{v,w\} \mid v,w \in Q\}$, then $G$ is \emph{bipartite (on bipartition $(P,Q)$)}. If, additionally, $\big\{\{v,w\} \mid v \in P, w \in Q\big\} = E(G)$, the graph $G$ is \emph{complete bipartite}.

For $k, \ell \in \mathbb{N}_0$, a \emph{$(k,\ell)$-biregular graph (on bipartition $(P,Q)$)}
  is a bipartite graph on bipartition $(P,Q)$ such that
  for every $v \in P$, it holds that $|N(v)| = k$, and
  for every $w \in Q$, it holds that $|N(w)| = \ell$. A \emph{biregular} graph is a graph $G$ for which there are $P, Q \subseteq V(G)$ and $k, \ell \in \mathbb{N}_0$ such that $G$ is $(k,\ell)$-biregular on bipartition $(P,Q)$.  

For a graph $G$ and a set $V' \subseteq V(G)$, we let $G[V']$ be the \emph{induced subgraph of $G$ on $V'$}, i.e.\ the subgraph of $G$ with vertex set $V'$ and edge set $E(G) \cap \{\{v,w\} \mid v,w \in V'\}$. We define $G - V' \coloneqq G[V(G) \setminus V']$. Furthermore, for vertex sets $V_1, V_2 \subseteq V(G)$, we denote by $G[V_1, V_2]$ the graph with vertex set $V_1 \cup V_2$ and edge set $E(G) \cap \{\{v_1,v_2\} \mid v_1 \in V_1, v_2 \in V_2\}$. 

A \emph{coloured graph} is a tuple $(G,\lambda)$, where $G$ is a graph and $\lambda \colon V(G) \rightarrow \mathcal{C}$ is a function that assigns colours (i.e.\ elements from a particular set $\mathcal{C}$) to the vertices. We interpret all graphs treated in this paper as coloured graphs and just write $G$ instead of $(G,\lambda)$ when $\lambda$ is clear from the context. If the colouring is not specified, we assume a monochromatic colouring, i.e.\ all vertices have the same colour.

For a coloured graph $G$ with colouring $\lambda$, a \emph{(vertex) colour class} of $G$ is a maximal set of vertices that all have the same $\lambda$-colour. Every graph colouring $\lambda$ induces a partition $\pi(\lambda)$ of $V(G)$ into the vertex colour classes with respect to $\lambda$.

\section{Colour Refinement}

Colour Refinement proceeds by iteratively refining a partition of the vertices of its input graph until the partition is stable with respect to the refinement criterion.

\begin{definition}[Colour Refinement]\label{def:k:dimensional:weisfeiler:lehman:refinement} Let~$\lambda \colon V(G) \rightarrow \mathcal{C}$ be a colouring of the vertices of a graph~$G$, where~$\mathcal{C}$ is some set of colours. The colouring computed by Colour Refinement on input $(G, \lambda)$ is defined recursively: we set $\chi^0_{G} \coloneqq \lambda$, i.e.\ the initial colouring is $\lambda$. For $i \in \mathbb{N}$, the colouring $\chi^i_{G}$ computed by Colour Refinement after $i$ iterations on $G$ is defined as $\chi^i_{G}(v) \coloneqq \big(\chi^{i-1}_G(v), \llcurly \chi^{i-1}_G(w) \mid w \in N(v) \rrcurly\big)$.
\end{definition}

That is, $\chi_G^i(v)$ consists of the colour of $v$ from the previous iteration as well as the multiset of colours of neighbors of $v$ from the previous iteration. It is not difficult to see that~$\pi(\chi^{i-1}_G) \succeq \pi(\chi^{i}_G)$ holds for every graph $G$ and every $i \in \mathbb{N}$. Therefore, there is a unique minimal integer $j$ such that $\pi(\chi^j_G) \equiv \pi(\chi^{j+1}_G)$. For this value $j$, we define the \emph{output} of Colour Refinement on input $G$ to be $\chi_{G} \coloneqq \chi^j_G$ and call $\chi_G$ and $\pi(\chi_G)$ \emph{the stable colouring} and \emph{the stable partition}, respectively, of $G$. Accordingly, executing~$i$ \emph{Colour Refinement iterations} on $G$ means computing the colouring~$\chi_G^{i}$. We call a graph~$G$ with colouring $\lambda$ and the induced partition $\pi(\lambda)$ \emph{stable} if~$\pi(\lambda) \equiv \pi(\chi_G)$. 
Note that for all $P,Q \in \pi(\chi_G)$ with $P \neq Q$, the graph $G[P]$ is regular and $G[P,Q]$ is biregular.

Colour Refinement can be used to check whether two given graphs $G$ and $G'$ are non-isomorphic by computing the stable colouring on the disjoint union of the two. If there is a colour~$C$ such that, in the stable colouring, the numbers of vertices of colour $C$ differ in $G$ and $G'$, they are non-isomorphic. However, even if they agree in every colour class size in the stable colouring, the graphs might not be isomorphic. It is not trivial to describe for which graphs this isomorphism test is always successful (see~\cite{DBLP:journals/cc/ArvindKRV17,kieschweisel15}).

\begin{notation}
We write $\WL_1(G)$ for the number of iterations of Colour Refinement on input $G$, that is, $\WL_1(G) = j$, where $j$ is the minimal integer for which $\pi(\chi^j_G) = \pi(\chi^{j+1}_G)$. Similarly, for $n \in \mathbb{N}$, we write $\WL_1(n)$ to denote the maximum number of iterations that Colour Refinement needs to reach stabilisation on an $n$-vertex graph. 
\end{notation}

We call every graph $G$ with $\WL_1(G) = |G|-1$ a \emph{long-refinement graph}.

\begin{fact}\label{fact:paths}
Let $G$ be an uncoloured path with $n$ vertices. Then $\WL_1(G) = \lfloor\frac{n-1}{2}\rfloor$.
\end{fact}

\begin{proof}[Proof sketch]
 In the first iteration, the two end vertices are distinguished from all others because they are the only ones with degree $1$. Then in each iteration, the information of being adjacent to a ``special'' vertex, i.e.\ the information about the distance to a vertex of degree 1, is propagated one step closer to the vertices in the centre of the path. This procedure takes $\lfloor \frac{n-1}{2} \rfloor$ iterations. 
\end{proof}

In 2015, Krebs and Verbitsky improved on the explicit linear lower bound for graphs of order $n$ given by Fact \ref{fact:paths} by constructing a family of pairs of graphs whose members of order $n$ can only be distinguished after $n - 8 \sqrt{n}$ Colour Refinement iterations (see \cite[Theorem 4.6]{krever15}). Hence, since for a set $\{v_1, \dots, v_n\} \eqqcolon S$ and partitions $\pi_1, \dots, \pi_\ell$ of $S$ that satisfy
\[\pi_1 \succneqq \pi_2 \succneqq \dots \succneqq \big\{\{v_1\}, \dots \{v_n\}\big\} = \pi_{\ell},\]
 it holds that $\ell \leq n-1$, we obtain the following corollary.

\begin{corollary}
For every $n \in \mathbb{N}$, it holds that $n - 8 \sqrt n \leq \WL_1(n) \leq n-1$.
\end{corollary}

It has remained open whether any of the two bounds is tight. In preliminary research conducted together with Gödicke and Schweitzer, towards improving the lower bound, the first author took up an approach to reverse-engineer the splitting of colour classes. Gödicke's implementation of those split procedures led to the following result.

\begin{theorem}[\cite{Goedicke}]\label{thm:goedicke}
For every $n \in \{1,10,11,12\}$, it holds that $\WL_1(n) = n-1$. For $n \in [2,9]$, it holds that $\WL_1(n) < n-1$.
\end{theorem}

Unfortunately, due to computational exhaustion, it was not possible to test for larger graph sizes. Also, the obtained graphs do not exhibit any structural properties that would lend themselves for a generalisation in order to obtain larger graphs.

Using a fast implementation of Colour Refinement, we could verify that there are exactly 16 long-refinement graphs of order 10, 24 long-refinement graphs of order 11, 32 of order 12, and 36 of order 13. However, again, with simple brute-force approaches, we could not go beyond those numbers exhaustively.

\section{Compact Representations of Long-Refinement Graphs}\label{sec:encodings}

In the light of the previous section, the question whether the lower bound obtained by Krebs and Verbitsky is asymptotically tight has remained open. With the brute-force approach, it becomes infeasible to test all graphs of orders much larger than 10 exhaustively for their number of Colour Refinement iterations until stabilisation. Still, knowing that there exist long-refinement graphs, it is natural to ask whether the ones presented in \cite{Goedicke} are exceptions or whether there are infinitely many such graphs. In this section, we show that the latter is the case.

When the input is a coloured graph with at least two vertex colours, the initial partition already has two elements. Hence, all long-refinement graphs are monochromatic. Therefore, in the following, all initial input graphs are considered to be monochromatic.

\begin{proposition}\label{prop:max:iterations}
 Let $G$ be a graph and let $n \coloneqq |G|$. If there exists an $i \in \mathbb{N}_0$ such that
 $|\{\chi^{i+1}_{G}(v) \mid v \in V(G)\}| - |\{\chi^i_{G}(v) \mid v \in V(G)\}| \geq 2$ holds, then $G$ is not a long-refinement graph.
\end{proposition}

\begin{proof}
Every pair of partitions $\pi, \pi'$ with $\pi \succneqq \pi'$ satisfies $|\pi'| \geq |\pi| + 1$. Thus, every sequence of partitions of the form
\[\pi_1 \coloneqq \big\{\{v_1, \dots, v_n\}\big\} \succneqq \pi_2 \succneqq \dots \succneqq \big\{\{v_1\}, \dots \{v_n\}\big\} \eqqcolon \pi_n\] 
must satisfy $|\pi_i| = |\pi_{i-1}| + 1$ for all $i \in [2,n]$.
\end{proof}

The proposition implies that in order to find long-refinement graphs, we have to look for graphs in which, in every Colour Refinement iteration, only one additional colour class appears. That is, in each iteration, only one colour class is split and the splitting creates exactly two new colour classes. 

\begin{corollary}\label{cor:two:degrees}
Let $G$ be a long-refinement graph with at least two vertices. Then there exist $d_1, d_2 \in \mathbb{N}_0$ with $d_1 \neq d_2$ and such that $\deg(G) = \{d_1,d_2\}$.
\end{corollary}

\begin{proof}
This is a direct consequence of Proposition \ref{prop:max:iterations}: every (monochromatic) regular graph $G$ satisfies $\WL_1(G)= 0$ and if there were more than two vertex degrees present in $G$, we would have $|\{\chi^{1}_{G}(v) \mid v \in V(G)\}| - |\{\chi^0_{G}(v) \mid v \in V(G)\}| \geq 3 - 1 \geq 2$.
\end{proof}

We can thus restrict ourselves to graphs with exactly two vertex degrees. 

\begin{notation}
For a graph $G$ and $i \in \mathbb{N}_0$, we let $\pi^i_G$ denote the partition induced by $\chi^i_{G}$ on $V(G)$, i.e.\ after $i$ Colour Refinement iterations on $G$. If $G$ is clear from the context, we omit it in the expression.
\end{notation}

As a result of the regularity conditions that must hold for the graph $G[V_1,V_2]$, we make the following observation. It implies that, in a long-refinement graph, to determine the class $C$ that is split in iteration $i$, it suffices to consider the neighbourhood of an arbitrary class obtained in the preceding iteration.

\begin{lemma}\label{lem:not:biregular}
 Let $G$ be a graph. Suppose there are $i \in \mathbb{N}$ and $C_1,C_2,C'$ with $\pi^{i} \setminus \pi^{i-1} = \{C_1,C_2\}$ and $C' \in \pi^{i} \setminus \pi^{i+1}$. Then there are vertices $v'_1, v'_2 \in C'$ such that $|N(v'_1) \cap C_1| \neq |N(v'_2) \cap C_1|$. 
\end{lemma}

\begin{proof}
Note that there must be a $C \in \pi^{i-1} \setminus \pi^i$ with $C_1 \cup C_2 = C$. Since $C' \in \pi^i$ and $C \in \pi^{i-1}$, there is a $d \in \mathbb{N}_0$ such that for every $v \in C'$, it holds that $d = |N(v) \cap C|$. Since $\{C_1,C_2\} = \pi^{i} \setminus \pi^{i-1}$ and $C' \notin \pi^{i+1}$, there are vertices $v'_1, v'_2 \in C'$ such that $|N(v'_1) \cap C_1| \neq |N(v'_2) \cap C_1|$ or $|N(v'_1) \cap C_2| \neq |N(v'_2) \cap C_2|$. In the first case, we are done. In the second case, we obtain $|N(v'_1) \cap C_1| = d - |N(v'_1) \cap C_2| \neq d - |N(v'_2) \cap C_2| = |N(v'_2) \cap C_1|$.
\end{proof}

Note that the validity of the lemma depends on the assumption $\{C_1,C_2\} = \pi^{i} \setminus \pi^{i-1}$, which by Proposition \ref{prop:max:iterations} is always fulfilled in long-refinement graphs as long as $\pi^{i-1} \not\equiv \pi^{i}$.

\begin{corollary}
No graph with more than one connected component is a long-refinement graph.
\end{corollary}

\begin{proof}
Since the refinement process takes place in parallel in each connected component, $\WL_1(G)$ is the maximum of all $WL_1(H)$ for the connected components $H$ of $G$.
\end{proof}

We can therefore restrict ourselves to connected graphs. The only connected graphs $G$ with $\deg(G) = \{1,2\}$ are paths and, by Fact \ref{fact:paths}, they are not long-refinement graphs. Thus, the smallest degree pairs for a search for candidates are $\{1,3\}$ and $\{2,3\}$.

\begin{lemma}\label{lem:degree:reduction}
Let $G$ be a long-refinement graph. Then $|\{v \in V(G) \mid \deg(v) = 1\}| \leq 2$. 
\end{lemma}

\begin{proof}
Suppose the lemma does not hold. Let $G$ be a long-refinement graph with at least three vertices of degree 1. Consider the execution of Colour Refinement on input $G$ and let $n \coloneqq |G|$. In $\pi^1$, there are two vertex colour classes, namely a class $V_1$ containing the vertices of degree 1 and a class $V_d$ containing the vertices of the second vertex degree $d \neq 1$. 

Suppose that $|V_1| \geq 2$. The class $V_1$ is not split before $N(V_1)$ has been split. Thus, consider the iteration $j$ after which $N(V_1)$ has been subdivided into two classes $W_1$ and $W_2$. This induces the splitting of $V_1$ into $N(W_1) \cap V_1$ and $N(W_2) \cap V_1$, which by Proposition \ref{prop:max:iterations} implies in particular that for all pairs of partition classes $C, C' \in \pi^{j}$ with $C \cap V_1 = \emptyset = C' \cap V_1$, the graph induced between the two classes is biregular. Therefore, however, now for every pair of classes $C, C' \in \pi^{j+1}$, the graph $G[C,C']$ is biregular and thus, the partition is equitable. Hence, $j = n-2$, i.e.\ the splitting of $V_1$ must happen in the $(n-1)$-st iteration. In particular, $N(W_1) \cap V_1$ and $N(W_2) \cap V_1$ must be singletons, i.e.\ $|V_1| = 2$. 
\end{proof}

\begin{table}[ht]
  \caption{Adjacency lists of long-refinement graphs $G$ with $\deg(G) = \{1,5\}$ (left) and $\deg(G) = \{1,3\}$ (right), respectively.}\label{t1}
  \vspace{-1ex}
  \begin{center}
    \begin{tabular}[t]{|r|l|}
      \hline 
      \rule[-1ex]{0pt}{2.5ex} $v$ & $N(v)$ \\ 
      \hlinewd{1.5pt} 
      \rule[-1ex]{0pt}{2.5ex}~0 &~1 \\ 
      \hline 
      \rule[-1ex]{0pt}{2.5ex} 1 &~0,2,3,4,5 \\ 
      \hline 
      \rule[-1ex]{0pt}{2.5ex} 2 &~1,3,5,7,10 \\ 
      \hline 
      \rule[-1ex]{0pt}{2.5ex} 3 &~1,2,4,6,10 \\ 
      \hline 
      \rule[-1ex]{0pt}{2.5ex} 4 &~1,3,5,9,11 \\ 
      \hline 
      \rule[-1ex]{0pt}{2.5ex} 5 &~1,2,4,8,11 \\ 
      \hline 
      \end{tabular}
      \hspace{1ex}
      \begin{tabular}[t]{|r|l|}
      \hline 
      \rule[-1ex]{0pt}{2.5ex} $v$ & $N(v)$ \\ 
      \hlinewd{1.5pt} 
      \rule[-1ex]{0pt}{2.5ex} 6 &~3,7,8,9,11  \\ 
      \hline 
      \rule[-1ex]{0pt}{2.5ex} 7 &~2,6,8,9,10  \\ 
      \hline 
      \rule[-1ex]{0pt}{2.5ex} 8 &~5,6,7,10,11  \\ 
      \hline 
      \rule[-1ex]{0pt}{2.5ex} 9 &~4,6,7,10,11  \\ 
      \hline 
      \rule[-1ex]{0pt}{2.5ex} 10 &~2,3,7,8,9  \\ 
      \hline 
      \rule[-1ex]{0pt}{2.5ex} 11 &~4,5,6,8,9  \\ 
      \hline 
    \end{tabular}
\hspace{1.8cm}
    \begin{tabular}[t]{|r|l|}
      \hline 
      \rule[-1ex]{0pt}{2.5ex} $v$ & $N(v)$ \\ 
      \hlinewd{1.5pt} 
      \rule[-1ex]{0pt}{2.5ex}~0 &~1 \\ 
      \hline 
      \rule[-1ex]{0pt}{2.5ex} 1 &~0,2,3 \\ 
      \hline 
      \rule[-1ex]{0pt}{2.5ex} 2 &~1,11,13 \\ 
      \hline 
      \rule[-1ex]{0pt}{2.5ex} 3 &~1,10,12 \\ 
      \hline 
      \rule[-1ex]{0pt}{2.5ex} 4 &~5,7,10 \\ 
      \hline 
      \rule[-1ex]{0pt}{2.5ex} 5 &~4,6,10 \\ 
      \hline 
      \rule[-1ex]{0pt}{2.5ex} 6 &~5,9,11  \\ 
      \hline 
      \end{tabular}
      \hspace{1ex}
      \begin{tabular}[t]{|r|l|}
      \hline 
      \rule[-1ex]{0pt}{2.5ex} $v$ & $N(v)$ \\ 
      \hlinewd{1.5pt}       
      \rule[-1ex]{0pt}{2.5ex} 7 &~4,8,11  \\ 
      \hline 
      \rule[-1ex]{0pt}{2.5ex} 8 &~7,9,13  \\ 
      \hline 
      \rule[-1ex]{0pt}{2.5ex} 9 &~6,8,12  \\ 
      \hline 
      \rule[-1ex]{0pt}{2.5ex} 10 &~3,4,5  \\ 
      \hline 
      \rule[-1ex]{0pt}{2.5ex} 11 &~2,6,7  \\ 
      \hline
      \rule[-1ex]{0pt}{2.5ex} 12 &~3,9,13  \\ 
      \hline 
      \rule[-1ex]{0pt}{2.5ex} 13 &~2,8,12  \\ 
      \hline  
    \end{tabular}
  \end{center}
\end{table}

Table \ref{t1} displays the adjacency lists of two long-refinement graphs on 12 and 14 vertices, respectively, which each have exactly one vertex of degree 1.

The lemma allows us to reduce the decision problem whether there are infinitely many long-refinement graph with degrees in $\{1,2,3\}$ to the question whether there are such families with degrees in $\{2,3\}$. 

\begin{corollary}
If there is a long-refinement graph $G$ with $\deg(G) = \{1,3\}$, then there is also a long-refinement graph $\hat G$ with $\deg(\hat G) = \{2,3\}$ and $|\hat G| \in \{|G|-1, |G|\}$.
\end{corollary}

\begin{proof}
Let $G$ be a long-refinement graph with $\deg(G) = \{1,3\}$. Then $\pi^1 = \{V_1,V_3\}$, where $V_1 = \{v \in V(G) \mid \deg(v) = 1\}$ and $V_3 = \{v \in V(G) \mid \deg(v) = 3\}$. By Lemma \ref{lem:degree:reduction}, it holds that $|V_1| \in \{1,2\}$.

First suppose $|V_1| = 2$. Consider the graph $\hat G$ with $V(\hat G) = V(G)$ and $E(\hat G) = E(G) \cup \{V_1\}$, i.e.\ obtained from $G$ by inserting an edge between the two vertices in $V_1$. In the following, we identify the vertices of $\hat G$ with their counterparts in $G$. For $i \in \mathbb{N}_0$, let $\hat\pi^i$ be the partition of $V(\hat G)$ induced by $\chi^i_{\hat G}$. Let $n \coloneqq |G|$. Then, for $i \in [0,n-1]$, it holds that 
\[\hat\pi^i = \pi^i.
\] 
This follows from $\hat G - V_1 = G - V_1$ and $\hat G[V_1, N(V_1)] = G[V_1,N(V_1)]$, the regularity of $\hat G[V_1]$ and that there is only one way to split $V_1$, which results in two singletons. In particular, it holds that $\WL_1(\hat G) = \WL_1(G) = |\hat G|-1$.  

Now suppose $|V_1| = 1$. In $\pi^1$, there are only the two partition classes $V_1$ and $V_3$. In $\pi^2$, the set $V_3$ is subdivided into the singleton $N(V_1)$ and $V_3 \setminus N(V_1)$. Define $\hat G \coloneqq G - V_1$ and again, for $i \in \mathbb{N}_0$, let $\hat\pi^i$ be the partition of $V(\hat G)$ induced by $\chi^i_{\hat G}$. Then $\hat \pi^1 = \{N(V_1), V_3 \setminus N(V_1)\} = \pi^2 \setminus \{V_1\}$ and, more generally, for $i \in \mathbb{N}$, we obtain $\hat \pi^i = \pi^{i+1} \setminus \{V_1\}$. This can be deduced from the equality $\hat G - V_1 = G - V_1$. Thus, $\WL_1(\hat G) = (n-1)-1 = |\hat G| - 1$.
\end{proof}

With the help of the tool \texttt{Nauty} \cite{mck81}, our quest for long-refinement graphs was successful. We tested exhaustively up to order 13. To render the search for larger long-refinement graphs tractable, we imposed further conditions. Restricting the degrees to $\{2,3\}$, it was possible to test for graphs up to order 64. Altogether, we found graphs $G$ with $n-1$ Colour Refinement iterations, where $n = |G|$, for all even $n \in [10,64] \setminus \{24,30,42,48,60\}$ and for all odd $n \in [11,63] \setminus \{21,27,39,45,57,63\}$.\footnote{We exclude the case $n=10$ in the following analysis since, as our computational results have shown, although long-refinement graphs of order 10 do exist, none of them has vertex degrees 2 and 3.}

In the following, in order to generalise the results to bigger graph sizes, we analyse the obtained graphs. Among our computational results, the even-size graphs $G$ with vertex degrees 2 and 3 have the following property in common: there is an iteration $j$ such that for every $C \in \pi^j_{G}$, it holds that $|C| = 2$. That is, with respect to their assigned colours, the vertices remain in pairs until there are no larger colour classes left. Then the first such pair is split into singletons, which must induce a splitting of another pair, and so on, until the discrete partition is obtained. (Similar statements hold for the odd-size graphs, but are more technical.) In the following, a \emph{pair} is a set of two vertices which occurs as a colour class during the execution of Colour Refinement. That is, vertices $v,v'$ form a pair if and only if $\{v,v'\}$ is an element of $\pi^i$ for some $i \in \mathbb{N}_0$.

As just argued, there is a splitting order on the pairs, i.e.\ a linear order $\prec$ induced by the order in which pairs are split into singletons. We now examine the possible connections between pairs. 

From now on, we make the following assumption.

\begin{assumption}\label{ass:max:iterations}
$G$ is a long-refinement graph with $\deg(G) = \{2,3\}$ and such that there is an $i \in \mathbb{N}_0$ for which $\pi^i$ contains only pairs. Let $\prec$ be the splitting order of these pairs.
\end{assumption}

We call pairs $P_1,P_2 \subseteq V(G)$ \emph{successive} if $P_2$ is the successor of $P_1$ with respect to $\prec$. Note that for successive pairs $P_1$, $P_2$, in the graph $G[P_1, P_2]$, every $v_2 \in P_2$ must have the same number of neighbours in $P_1$, otherwise it would hold that $P_2 \prec P_1$. By a simple case analysis, together with an application of Lemma \ref{lem:not:biregular}, this rules out all connections but matchings for successive pairs.

\begin{corollary}\label{cor:succ:pairs:match}
 Let $P_1$ and $P_2$ be successive pairs. Then $G[P_1,P_2]$ is a matching.
\end{corollary}

Towards a compact representation of the graphs, we further examine the connections between pairs $P_1$ and $P_2$ with $S(P_1) \prec P_2$, where $S(P_1)$ is the successor of $P_1$ with respect to $\prec$. 

\begin{lemma}\label{lem:edges:nonsucc}
Let $P_1$ be a pair. Then exactly one of the following holds.

\begin{itemize}
  \item $P_1 \neq \min(\prec)$ and for every pair $P_2$ with $S(P_1) \prec P_2$, it holds that $E(G[P_1, P_2]) = \emptyset$.
  \item $P_1 = \min(\prec)$ and there are exactly two choices $P_2, P'_2$ for a pair $P'$ with $S(P_1) \prec P'$ such that $E(G[P_1, P']) \neq \emptyset$. Furthermore, there is a vertex $v_1 \in P_1$ such that $G[\{v_1\},P_2]$ and $G[P_1 \setminus \{v_1\},P'_2]$ are complete bipartite and $E(G[\{v_1\},P'_2]) = E(G[P_1 \setminus \{v_1\},P_2]) = \emptyset$.
\end{itemize}
\end{lemma}

\begin{proof}
Suppose $P_1 \neq \min(\prec)$. If $P_1 = \max(\prec)$, the statement trivially holds. Otherwise, by Corollary \ref{cor:succ:pairs:match}, every vertex $v_1 \in P_1$ has exactly one neighbour in $S(P_1)$ and exactly one neighbour in the predecessor of $P_1$, i.e.\ in the unique pair $A(P_1)$ such that $P_1 = S(A(P_1))$. Thus, due to the degree restrictions, $v_1$ can have at most one additional neighbour in a pair $P'$ with $P_1 \prec P'$ and $P' \neq S(P_1)$. However, if $v_1$ had a neighbour in such a $P'$, the graph $G[\{v_1\},P']$ would not be biregular, implying that $P' = S(P_1)$, a contradiction. Therefore, $N(v_1) \subseteq A(P_1) \cup P_1 \cup S(P_1)$ and thus, $N(P_1) \subseteq A(P_1) \cup S(P_1)$. In particular, for every pair $P_2$ with $P_1 \prec P_2$ and $P_2 \neq S(P_1)$, it holds that $E(G[P_1, P_2]) = \emptyset$.

Now suppose that $P_1 = \min (\prec)$. Since the splitting of $P_1$ must be induced by a splitting of a union of two pairs and $G[P_1,S(P_1)]$ is biregular and $G[P_1]$ is regular, we cannot have $N(P_1) \subseteq S(P_1)$. Thus, there is a pair $P_2$ with $S(P_1) \prec P_2$ and such that $E(G[P_1,P_2]) \neq \emptyset$. Let $v_1 \in P_1$ be a vertex with $N(v_1) \cap P_2 \neq \emptyset$. Then $P_2 \subseteq N(v_1)$, otherwise $P_2 = S(P_1)$. Thus, $G[\{v_1\},P_2]$ is complete bipartite. Therefore and due to the degree restrictions, $v_1$ has exactly three neighbours: one in $S(P_1)$ and two in $P_2$. In particular, for every pair $P'_2$ with $P_2 \neq P'_2 \neq S(P_1)$, it holds that $E(G[\{v_1\},P'_2]) = \emptyset$.

Let $v'_1 \neq v_1$ be the second vertex in $P_1$. Since the splitting of $P_1$ induces the splitting of $S(P_1)$, by Proposition \ref{prop:max:iterations}, for every pair $P'$ with $P_1 \neq P' \neq S(P_1)$, the graph $G[\{v'_1\},P']$ must be biregular, i.e.\ either empty or complete bipartite.

Moreover, since $\deg(v_1) = 3$, also $\deg(v'_1) = 3$. By Corollary \ref{cor:succ:pairs:match}, it holds that $|N(v'_1) \cap S(P_1)| = 1$. Therefore, there is exactly one pair $P'_2$ such that $G[\{v'_1\},P'_2]$ is complete bipartite and for all other pairs $P'$ with $P_1 \neq P' \neq S(P_1)$, the graph $G[\{v'_1\},P']$ is empty. 

Suppose $P'_2 = P_2$. Choose $i$ such that $\pi^{i} \setminus \pi^{i+1} = \{P_1\}$. Then the unique element in $\pi^{i-1} \setminus \pi^{i}$ is a union of two pairs, whose splitting induces the splitting of $P_1$. However, $N(P_1) = S(P_1) \cup P_2$ and both graphs $G[P_1, S(P_1)]$ and $G[P_1,P_2]$ are biregular. 

Thus, $P'_2 \neq P_2$, which concludes the proof.
\end{proof}

Corollary \ref{cor:succ:pairs:match} and Lemma \ref{lem:edges:nonsucc} characterise $G[P_1,P_2]$ for all pairs $P_1 \neq P_2$. Thus, all additional edges must be between vertices from the same pair. Hence, we can use the following compact graphical representation to fully describe the graphs of order at least 12 that we found. As the set of nodes, we take the pairs. We order them according to $\prec$ and connect successive pairs with an edge representing the matching. If the two vertices of a pair are adjacent, we indicate this with a loop at the corresponding node. The only other type of connection between pairs is constituted by the edges from $\min(\prec)$ to two other pairs which form the last colour class of size 4, i.e.\ a colour class of size 4 in the partition $\pi^i$ for which $\pi^{i+1} \setminus \pi^{i+2} = \{\min(\prec)\}$. We indicate this type of edge with a dotted curve.

An example graph as well as the evolution of the colour classes computed by Colour Refinement on the graph is depicted in Figure \ref{fig:col:ref:max:iterations}. 

\begin{figure}[h]
\newcommand\vDist{.40cm}
\newcommand{\graphwidth}{0.22\textwidth}
\begin{tikzpicture}
\node[anchor=north west,text width=\textwidth-1pt,inner sep=0] (line1)
    {\includegraphics[page=1,width=\graphwidth]{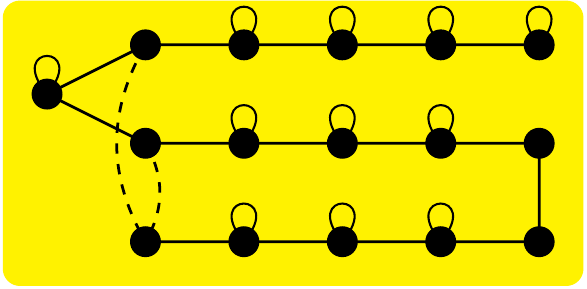} \hfill
     \includegraphics[page=2,width=\graphwidth]{./graph2-game.pdf} \hfill
     \includegraphics[page=3,width=\graphwidth]{./graph2-game.pdf} \hfill
     \includegraphics[page=4,width=\graphwidth]{./graph2-game.pdf}};
\draw ($(line1.south west)+(0,-\vDist/2)$) -- ($(line1.south east)+(0,-\vDist/2)$);
\node[anchor=north west,text width=\textwidth-1pt,inner sep=0] (line2) at ($(line1.south west)-(0,\vDist)$)
    {\includegraphics[page=5,width=\graphwidth]{./graph2-game.pdf} \hfill
     \includegraphics[page=6,width=\graphwidth]{./graph2-game.pdf} \hfill
     \includegraphics[page=7,width=\graphwidth]{./graph2-game.pdf} \hfill
     \includegraphics[page=8,width=\graphwidth]{./graph2-game.pdf}};
\draw ($(line2.south west)+(0,-\vDist/2)$) -- ($(line2.south east)+(0,-\vDist/2)$);
\node[anchor=north west,text width=\textwidth-1pt,inner sep=0] (line3) at ($(line2.south west)-(0,\vDist)$)
    {\includegraphics[page=9,width=\graphwidth]{./graph2-game.pdf} \hfill
     \includegraphics[page=10,width=\graphwidth]{./graph2-game.pdf} \hfill
     \includegraphics[page=11,width=\graphwidth]{./graph2-game.pdf} \hfill
     \includegraphics[page=12,width=\graphwidth]{./graph2-game.pdf}};
\draw ($(line3.south west)+(0,-\vDist/2)$) -- ($(line3.south east)+(0,-\vDist/2)$);
\node[anchor=north west,text width=\textwidth-1pt,inner sep=0] (line4) at ($(line3.south west)-(0,\vDist)$)
    {\includegraphics[page=13,width=\graphwidth]{./graph2-game.pdf} \hfill
     \includegraphics[page=14,width=\graphwidth]{./graph2-game.pdf} \hfill
     \includegraphics[page=15,width=\graphwidth]{./graph2-game.pdf} \hfill
     \includegraphics[page=16,width=\graphwidth]{./graph2-game.pdf}};
\end{tikzpicture}
\caption{Top left: A long-refinement graph $G$ on 32 vertices. The subsequent pictures show the partitions of $V(G)$ after the first 15 Colour Refinement iterations. There are 16 further iterations not depicted here, which consist in the splitting of the pairs into singletons.}\label{fig:col:ref:max:iterations}
\end{figure}

\begin{notation}\label{not:string}
Since $\prec$ is a linear order, we can also use a string representation to fully describe the graphs. For this, we introduce the following notation, letting $A(P)$ and $S(P)$ be the predecessor and successor of $P$, respectively, with respect to $\prec$. 
\begin{itemize}
  \item 0 represents a pair of vertices of degree 2.
  \item 1 represents a pair $P$ of vertices of degree 3 that is not the minimum of $\prec$ and for which $N(P) \subseteq A(P) \cup S(P)$. (This implies that $P \in E(G)$.)
  \item X represents a pair $P$ of vertices of degree 3 that is not the minimum of $\prec$ and for which $N(P) \not\subseteq A(P) \cup S(P)$. 
  \item S represents the minimum of $\prec$.
\end{itemize}
\end{notation}

Thus, by Lemma \ref{lem:edges:nonsucc}, there are exactly two pairs of type X, namely $P_2$ and $P'_2$ from the lemma. Now we can use the alphabet $\Sigma = \{0,1,\mathrm{S},\mathrm{X}\}$ and the order $\prec$ to encode the graphs as strings. The $i$-th letter of a string is the $i$-th element of $\prec$. Note that S is always a pair of non-adjacent vertices of degree 3 due to the degree restrictions. For example, the string representation for the graph in Figure \ref{fig:col:ref:max:iterations} is S11100111X1X1110.

Formally, for every $\ell \geq 3$ and every string $\Xi \colon [\ell] \rightarrow \{0,1,\mathrm{S},\mathrm{X}\}$ with $\Xi(1) = \mathrm{S}$ and $\Xi^{-1}(\mathrm{X}) = \{r,r'\}$ for some $r, r' \in [\ell]$ with $r < r'$, we define the corresponding graph $G \coloneqq G(\Xi)$ with $V(G) = \big\{v_{i,j} \, \big\vert \, i \in [\ell], j \in [2]\big\}$ and
\begin{align*}
E(G) ={} &\big\{\{v_{i,1},v_{i,2}\} \, \big\vert \, i \in [\ell], \Xi(i) = 1\big\} 
\cup{}\\&
\big\{\{v_{i,j},v_{i+1,j}\} \, \big\vert \, i \in [\ell-1], j \in [2]\big\} 
\cup{}\\& 
\big\{\{v_{r\mspace{0.9mu},j},v_{1,1}\} \, \big\vert \, j \in [2]\big\} 
\cup{}\\&
\big\{\{v_{r'\mspace{-5mu},j},v_{1,2}\} \, \big\vert \, j \in [2]\big\}.
\end{align*}

We use this encoding in the next section, which contains our main results.

\section{Infinite Families of Long-Refinement Graphs}\label{sec:long:ref:graphs}

In this section, we present infinite families of long-refinement graphs. We adapt them further to deduce that $\WL_1(n) \geq n-2$ holds for all $n \in \mathbb{N}_{\geq 10}$.

For $w \in \{0,1\}^*$, the notation $(w)^k$ abbreviates the $k$-fold concatenation of $w$. We let $1^k \coloneqq (1)^k$.

\begin{figure}[ht]
\newcommand\vDist{.40cm}
\newcommand\graphwidth{.3\textwidth}
\begin{tikzpicture}
\node[anchor=north west,text width=\textwidth-1pt,inner sep=0] (line1)
  {\includegraphics[page=1,width=\graphwidth]{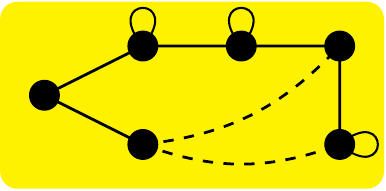} \hfill
   \includegraphics[page=2,width=\graphwidth]{./graph7-s011xx.pdf} \hfill
   \includegraphics[page=3,width=\graphwidth]{./graph7-s011xx.pdf}};
\draw ($(line1.south west)+(0,-\vDist/2)$) -- ($(line1.south east)+(0,-\vDist/2)$);
\node[anchor=north west,text width=\textwidth-1pt,inner sep=0] (line2)  at ($(line1.south west)-(0,\vDist)$)
  {\includegraphics[page=4,width=\graphwidth]{./graph7-s011xx.pdf} \hfill
   \includegraphics[page=5,width=\graphwidth]{./graph7-s011xx.pdf} \hfill
   \includegraphics[page=6,width=\graphwidth]{./graph7-s011xx.pdf}};
\end{tikzpicture}
\caption{A visualisation of the graph with string representation S011XX and the evolution of the colour classes in the first 5 Colour Refinement iterations on the graph.}\label{fig:s011xx}
\end{figure}

\begin{theorem}\label{thm:infinite:families}
 For every string $\Xi$ contained in the following sets, the graph $G(\Xi)$ is a long-refinement graph.
 \begin{itemize}
  \item $\{\mathrm{S011XX}\}$
  \item $\{\mathrm{S1^{\mathit{k}}001^{\mathit{k}}X1X1^{\mathit{k}}0} \mid k \in \mathbb{N}_0\}$
  \item $\{\mathrm{S1^{\mathit{k}}11001^{\mathit{k}}XX1^{\mathit{k}}0} \mid k \in \mathbb{N}_0\}$
  \item $\{\mathrm{S1^{\mathit{k}}0011^{\mathit{k}}XX1^{\mathit{k}}10} \mid k \in \mathbb{N}_0\}$
  \item $\{\mathrm{S011(011)^{\mathit{k}}00(110)^{\mathit{k}}XX(011)^{\mathit{k}}0} \mid k \in \mathbb{N}_0\}$
  \item $\{\mathrm{S(011)^{\mathit{k}}00(110)^{\mathit{k}}1X0X1(011)^{\mathit{k}}0} \mid k \in \mathbb{N}_0\}$
\end{itemize}
\end{theorem}

\begin{proof}
 Let $G \coloneqq G(\mathrm{S011XX})$ (cf.\ Figure \ref{fig:s011xx}). The vertices $v_{2,1}$ and $v_{2,2}$ are the only ones of degree 2. Thus,
 \begin{align*}
 \pi^1 &= \big\{\{v_{2,1}, v_{2,2}\}, V(G) \setminus \{v_{2,1}, v_{2,2}\}\big\},
 \\\pi^2 &= \big\{\{v_{2,1}, v_{2,2}\}, \{v_{i,j} \mid i \in \{1,3\}, j \in [2]\}, \{v_{i,j} \mid i \in [4,6], j \in [2]\}\big\}.
 \intertext{Then} 
 \pi^3 &= \big\{\{v_{1,1}, v_{1,2}\}, \{v_{2,1}, v_{2,2}\}, \{v_{3,1}, v_{3,2}\}, \{v_{i,j} \mid i \in [4,6], j \in [2]\}\big\},
 \intertext{since the vertices in the S-pair have no neighbours in $\{v_{i,j} \mid i \in \{1,3\}, j \in [2]\}$. Similarly,} 
 \pi^4 &= \big\{\{v_{1,1}, v_{1,2}\}, \{v_{2,1}, v_{2,2}\}, \{v_{3,1}, v_{3,2}\}, \{v_{4,1}, v_{4,2}\}, \{v_{i,j} \mid i \in [5,6], j \in [2]\}\big\},\\
  \pi^5 &= \big\{\{v_{i,j} \mid j \in [2]\} \mid i \in [6]\big\}.
  \end{align*}
 Now the splitting of the last colour class of size 4 into two X-pairs induces the splitting of the S-pair into singletons, which is propagated linearly according to $\prec$, adding 6 further iterations, thus summing up to 11 iterations.

 We now consider the various infinite families of graphs. The proofs for them work similarly by induction over $k$. Therefore, we only present the full detailed proof for the family $\{\mathrm{S1^{\mathit{k}}001^{\mathit{k}}X1X1^{\mathit{k}}0} \mid k \in \mathbb{N}_0\}$, which includes the graph from Figure \ref{fig:col:ref:max:iterations}. 

 For $k=0$, the graph $G_0 \coloneqq G(\mathrm{S00X1X0})$ has 14 vertices. It is easy to verify that it indeed takes 13 Colour Refinement iterations to stabilise. We sketch how Colour Refinement processes the graph: for this, for $i \in \mathbb{N}_0$, we let $\pi_0^i$ denote the partition of $V(G_0)$ induced by $\chi^i_{G_0}$, i.e.\ after $i$ iterations of Colour Refinement on $G_0$. First, vertices are assigned colours indicating their degrees. That is, 
 \begin{align*}
 \pi_0^1 = \big\{&\{v_{i,j} \mid i \in \{2,3,7\}, j \in [2]\}, \{v_{i,j} \mid i \in \{1,4,5,6\}, j \in [2]\}\big\}.
 \intertext{Now} 
 \pi_0^2 = \big\{&\{v_{i,j} \mid i \in \{2,3,7\}, j \in [2]\}, \{v_{i,j} \mid i \in \{1,4,6\}, j \in [2]\}, \{v_{5,i} \mid i \in [2]\}\big\},
 \intertext{since the vertices contained in the $1$-pair are not adjacent to vertices from $0$-pairs. Since no vertex contained in the S-pair is adjacent to any vertex from the $1$-pair, we obtain}
  \pi_0^3 = \big\{&\{v_{i,j} \mid i \in \{2,3,7\}, j \in [2]\}, \{v_{i,j} \mid i \in \{4,6\}, j \in [2]\}, \{v_{5,i} \mid j \in [2]\}, 
  \\&\{v_{1,j} \mid j \in [2]\}\big\}.
 \intertext{Furthermore,}
 \pi_0^4 = \big\{&\{v_{i,j} \mid i \in \{3,7\}, j \in [2]\}, \{v_{i,j} \mid i \in \{4,6\}, j \in [2]\}, \{v_{5,i} \mid j \in [2]\}, 
 \\&\{v_{1,j} \mid j \in [2]\}, \{v_{2,j} \mid j \in [2]\}\big\},
 \\
 \pi_0^5 = \big\{&\{v_{7,j} \mid j \in [2]\}, \{v_{i,j} \mid i \in \{4,6\}, j \in [2]\}, \{v_{5,i} \mid j \in [2]\}, 
 \\&\{v_{1,j} \mid j \in [2]\}, \{v_{2,j} \mid j \in [2]\}, \{v_{3,j} \mid j \in [2]\}\big\},
 \\
 \pi_0^6 = \big\{&\{v_{i,j} \mid j \in [2]\} \mid i \in [7]\big\},
 \end{align*}
 i.e.\ with respect to the order $\prec$ induced by the string representation, the first $0$-pair, the second $0$-pair and the first X-pair are separated from the others. 
 Once the two X-pairs form separate colour classes, this induces the splitting of S into two singletons, which is propagated linearly through the entire string, adding 7 further iterations, thus summing up to 13 iterations.

 For general $k \geq 1$, let $G_k \coloneqq G(\mathrm{S1^{\mathit{k}}001^{\mathit{k}}X1X1^{\mathit{k}}0})$. To count the iterations of Colour Refinement, we introduce some vocabulary for the pairs in $G_k$ (see also Figure \ref{fig:col:ref:max:iterations}). We let $V \coloneqq \{v_{i,j} \mid i \in [2,k+1] \cup [k+4,2k+3] \cup [2k+7,3k+6], j \in [2]\}$. 
 Note that $V$ is the set of vertices contained in the subgraphs corresponding to the substrings $1^k$ in the string representation. Furthermore, for all $i \in [k+2]$, we call the set $\{v_{i',j} \mid i' \in \{i,2k+5-i,2k+5+i\}, j \in [2]\}$ the \emph{$i$-th column} and denote it by $V_i$. The \emph{$0$-th column} is the set $\{v_{2k+5,j} \mid j \in [2]\}$. Thus, 
 \[V = \bigcup_{2 \leq i \leq k+1} V_i.\] 
 For every $j \in [2]$, the sets $\{v_{i,j} \mid i \in [1,k+2]\}$, $\{v_{i,j} \mid i \in [k+3,2k+4]\}$, and $\{v_{i,j} \mid i \in [2k+6,3k+7]\}$ are called \emph{rows}. In accordance with Figure \ref{fig:col:ref:max:iterations}, we fix an ordering on the rows: the \emph{first row} is $\{v_{i,1} \mid i \in [2k+6,3k+7]\}$, the \emph{second row} is $\{v_{i,2} \mid i \in [2k+6,3k+7]\}$, \dots, the \emph{sixth row} is $\{v_{i,2} \mid i \in [1,k+2]\}$. To be able to refer to the vertices in $V$ and the adjacent columns more easily, we relabel them: for $i \in [k+2], j \in [6]$, the vertex $w_{i,j}$ is defined to be the unique vertex in the $i$-th column and the $j$-th row.

 The following observation is the crucial insight for counting the iterations of Colour Refinement on $G_k$. We will use it to show that, informally stated, the subgraph $G_k[V]$ delays the propagation of the splitting of the colour classes in the remainder of the graph by $k$ iterations whenever the splitting of a colour class contained in $V_1$ or $V_{k+2}$ initiates a splitting of a colour class contained in $V$.

\begin{claim}
 Consider a colouring $\lambda$ of $G_k$ and its induced partition $\pi_k$ of $V(G_k)$. For $t \in \mathbb{N}_0$, let $\pi_k^t$ be the partition induced by $\chi^t_{G_k}$ on input $(G_k, \lambda)$. Suppose $G_k, \lambda, \pi_k$ satisfy the following conditions.
 \begin{enumerate}
  \item There exist $\ell \in [6]$ and $I_1, \dots, I_\ell \subseteq [6]$ such that $\bigcup_{i \in [\ell]} I_i = [6]$ and $I_i \cap I_j = \emptyset$ for $1 \leq i < j \leq \ell$ and for every $i \in [\ell]$, it holds that $\{w_{k+2,j'} \mid j' \in I_i\} \in \pi_k^0$. That is, $V_{k+2}$ is a union of colour classes with respect to $\lambda$.\label{prereq:1}
  \item $\pi_k^1 = \big\{\{w_{k+1,j'} \mid j' \in I_i\} \, \big\vert \, i \in [\ell]\big\}  \cup \big\{C \setminus V_{k+1} \, \big\vert \, C \in \pi_k, C \setminus V_{k+1} \neq \emptyset\big\}$.\label{prereq:2}
  \item For all $C,C' \subseteq V_{k+1}$ with $C,C' \in \pi_k^1$, the graph $G_k[C]$ is regular and $G_k[C,C']$ is biregular.\label{prereq:3}
\end{enumerate}
Then for every $t \in [k]$, it holds that
\begin{align*}
\pi_k^{t} = &\bigcup_{i' \in [t]} \big\{\{w_{k+2-i',j'} \mid j' \in I_i\} \, \big\vert \, i \in [\ell]\big\} \cup {}
\\&\bigg\{C \setminus \Big(\bigcup_{i' \in [t]} V_{k+2-i'}\big) \, \bigg\vert \, C \in \pi_k^0, C \setminus \Big(\bigcup_{i' \in [t]} V_{k+2-i'}\Big) \neq \emptyset\bigg\}.
\end{align*}
\end{claim}

\begin{claimproof}
We show the claim via induction. For $t = 1$, the statement is exactly the second item from the assumptions. For the inductive step, suppose the statement holds for all $t' \leq t$ for some $t \in [k-1]$. We show that it also holds for $t+1$. 

Note that the right-hand side of the equation is a partition of $V(G_k)$. Thus, it suffices to show ``$\supseteq$'', i.e.\ that the right-hand side is contained in the left-hand side of the equation.

Since $V_{k+2}$ is a union of elements of $\pi_k^0$, it is also a union of elements of $\pi_k^t$. Thus, by the induction hypothesis, 
\begin{align}\label{eq:first:column}
\{C \mid C \in \pi_k^0, C \cap V_{k+2} \neq \emptyset\} &= \{C \mid C \in \pi_k^0, C \subseteq V_{k+2}\}\nonumber 
\\&\subseteq \bigg\{C \mspace{-2mu}\setminus\mspace{-3mu} \Big(\mspace{-4mu}\bigcup_{i' \in [t]} \mspace{-3mu} V_{k+2-i'}\Big) \, \bigg\vert \, C \in \pi_k^0, C \mspace{-2mu}\setminus\mspace{-3mu} \Big(\mspace{-4mu}\bigcup_{i' \in [t]} \mspace{-3mu} V_{k+2-i'}\mspace{-2mu}\Big) \mspace{-2mu}\neq \mspace{-2mu}\emptyset\bigg\}\nonumber 
\\&\subseteq \pi_k^t.
\end{align}
Similarly, using the induction hypothesis for $t-1$, all other $V_i$ with $i \geq k+2-(t-1)$ are unions of elements of $\pi_k^{t-1}$ (if $t=1$, this holds trivially). Thus, 
\begin{align*}
\bigg\{C \, \bigg\vert \, C \in \pi_k^{t-1}, C \cap \bigcup_{i' \in [t-1]} V_{k+2-i'} \neq \emptyset\bigg\} &= \bigg\{C \, \bigg\vert \, C \in \pi_k^{t-1}, C \subseteq \bigcup_{i' \in [t-1]} V_{k+2-i'}\bigg\} 
\\&= \bigcup_{i' \in [t-1]} \big\{\{w_{k+2-i',j'} \mid j' \in I_i\} \, \big\vert \, i \in [\ell]\big\}
\\&\subseteq \pi_k^t.
\end{align*}

This means that all elements from $\pi_k^{t-1}$ that have a non-empty intersection with the set $\bigcup_{i' \in [0,t-1]} V_{k+2-i'}$ are also present in $\pi_k^t$. Therefore, for all $C, C' \in \pi_k^t$ with $C \cap \bigcup_{i' \in [0,t-1]} V_{k+2-i'} \neq \emptyset \neq C' \cap \bigcup_{i' \in [0,t-1]} V_{k+2-i'}$, the graph $G_k[C]$ must be regular and $G_k[C,C']$ must be biregular. (Otherwise, at least one of these classes would have been split in the $t$-th iteration.)

Actually, this also holds when relaxing the restriction for $C'$ to have a non-empty intersection with $\bigcup_{i' \in [0,t]} V_{k+2-i'}$. Indeed, $G_k[V_{k+2-t},V_{k+2-(t-1)}]$ is a matching between vertices contained in equal rows and, by the induction hypothesis for $t$, it holds that $\{C \mid C \in \pi_k^t, C \subseteq V_{k+2-(t-i')}\} = \big\{\{w_{k+2-(t-i'),j'} \mid j' \in I_i\} \mid i \in [\ell]\big\}$ for $i' \in \{0,1\}$. Thus, for all $C,C' \in \pi_k^t$ with $C \subseteq V_{k+2-(t-1)}$ and $C' \subseteq  V_{k+2-t}$, the graph $G_k[C,C']$ is either a perfect matching or empty. In particular, it is biregular. 

Note that there are no edges between vertices from columns $V_i$, $V_{i'}$ with $|i-i'| \geq 2$. Hence, in fact, for every $C \in \pi_k^t$ with $C\cap \bigcup_{i' \in [0,t-1]} V_{k+2-i'} \neq \emptyset$ and for all $C' \in \pi_k^t$, the graph $G_k[C]$ is regular and $G_k[C,C']$ is biregular. Thus, every $C \in \pi_k^t$ with $C\cap \bigcup_{i' \in [0,t-1]} V_{k+2-i'} \neq \emptyset$ is present in $\pi_k^{t+1}$. This shows that
\begin{equation}\label{eq:hypothesis:columns}
\bigcup_{i' \in [t-1]}\big\{\{w_{k+2-i',j'} \mid j' \in I_i\} \, \big\vert \, i \in [\ell]\big\} \subseteq \pi_k^{t+1}.
\end{equation} 
and, using \eqref{eq:first:column}, that
\begin{equation}\label{eq:first:column:outside}
\big\{C \, \big\vert \, C \in \pi_k^0, C \cap V_{k+2} \neq \emptyset \big\} \subseteq \pi_k^{t+1}.
\end{equation} 

We have seen that for $C,C' \in \pi_k^t$ with $C \subseteq \bigcup_{i' \in [0,t-1]} V_{k+2-i'}$ and $C' \subseteq V_{k+2-t}$, the graph $G[C]$ is regular and $G[C,C']$ is biregular. We now show that we can actually relax the location restriction for $C$ to $C \subseteq \bigcup_{i' \in [0,t]} V_{k+2-i'}$. For this, note that all subgraphs $G_k[V_i]$ with $i \in [2,k+1]$ have the same structure. That is, for all $r,s \in [6]$ and all $i,j \in [2,k+1]$, the vertices in $V_i$ in rows $r$ and $s$ are adjacent if and only if the corresponding ones are adjacent in $V_j$. Furthermore, by the induction assumption for $t$, it holds that $\{C \mid C \in \pi_k^t, C \subseteq V_{i'}\} = \big\{\{w_{i',j'} \mid j' \in I_i\} \mid i \in [\ell]\big\}$ for $i' \in \{k+2-t,k+1\}$. Thus, by Conditions \eqref{prereq:2} and \eqref{prereq:3} from the prerequisites of the claim, for all $C,C' \subseteq V_{k+2-t}$ with $C,C' \in \pi_k^t$, the graph $G_k[C]$ is regular and $G_k[C,C']$ is biregular.

In order to determine the colour classes of $\chi^{t+1}_{G_k}$ contained in $V_{k+2-t}$, we still need to analyse the structure of the graph $G_k[V_{k+1-t},V_{k+2-t}]$ with respect to $\pi_k^t$. To this end, for $i \in [t+1]$ and $j \in [0,k+2]$, set $M^i_j \coloneqq \{C \cap V_{j} \mid C \in \pi_k^{i}, C \cap V_{j} \neq \emptyset\}$. That is, $M^i_j$ is the partition of $V_{j}$ induced by $\pi_k^i$. 

The graph $G_k[V_{k+2-(t+1)},V_{k+2-t}]$ is a matching between vertices contained in equal rows and furthermore, by the induction assumption, we know $\{C \mid C \in \pi_k^t, C \subseteq V_{k+2-t}\} = \big\{\{w_{k+2-t,j'} \mid j' \in I_i\} \mid i \in [\ell]\big\}$. Therefore, 
\[M^{t-1}_{k+1-t} \succeq \big\{\{w_{k+1-t,j'} \mid j' \in I_i\} \mid i \in [\ell]\big\}.\]
However, the induction assumption yields $M^{i-1}_j =M^i_j$ for all $i \in [t], j \in [0,k+1-t]$. In particular, the partition of $V_{k+1-t}$ induced by $\pi_k^{t}$ is not strictly finer than the one induced by $\pi_k^{t-1}$. Thus, 
\begin{equation}\label{eq:coarse:columns}
M^{t}_{k+1-t} \succeq \big\{\{w_{k+1-t,j'} \mid j' \in I_i\} \mid i \in [\ell]\big\}.
\end{equation}
Therefore, again using the induction hypothesis for $t$, we obtain
\begin{equation}\label{eq:old:column}
\big\{\{w_{k+2-t,j'} \mid j' \in I_i\} \, \big\vert \, i \in [\ell]\big\} \subseteq \pi_k^{t+1}.
\end{equation}
Since $M^{t-1}_{k+1-t} = M^t_{k+1-t}$ and $M^{t-1}_{k-t} = M^t_{k-t}$, from \eqref{eq:coarse:columns}, we know $M^{t}_{k-t} \succeq M^{t}_{k+1-t} \succeq M^{t}_{k+2-t}$. It follows that
\[M^{t+1}_{k+1-t} = \big\{\{w_{k+1-t,j'} \mid j' \in I_i\} \mid i \in [\ell]\big\}. 
\]
Moreover, since $V_{k+1-t} \subseteq N(V_{k+2-t})$ and $V_{k+2-t}$ is a union of elements of $\pi_k^t$, the column $V_{k+2-(t+1)} = V_{k+1-t}$ must be a union of elements of $\pi_k^{t+1}$. Hence,
\begin{equation}\label{eq:new:column}
\big\{\{w_{k+2-(t+1),j'} \mid j' \in I_i\} \, \big\vert \, i \in [\ell]\big\} \subseteq \pi_k^{t+1}.
\end{equation} 

Putting \eqref{eq:hypothesis:columns}, \eqref{eq:old:column}, and \eqref{eq:new:column} together, we get
\begin{equation}\label{eq:inner:columns}
\bigcup_{i' \in [t+1]}\big\{\{w_{k+2-i',j'} \mid j' \in I_i\} \, \big\vert \, i \in [\ell]\big\} \subseteq \pi_k^{t+1}.
\end{equation}

Recall that by the induction hypothesis, it holds that $\big\{C \setminus (\bigcup_{i' \in [t]} V_{k+2-i'}) \, \big\vert \, C \in \pi_k^0, C \setminus (\bigcup_{i' \in [t]} V_{k+2-i'}) \neq \emptyset\big\} \subseteq \pi_k^t$. Thus, for every $C' \in \pi_k^t$ and for all vertices $u, v \in \bigcup_{i=0}^{k-t} V_i$ for which there is a $C \in \pi_k^0$ with $u,v \in C$, it holds that $|N(u) \cap C'| = |N(v) \cap C'|$. (To see this, recall that $N(u), N(v) \subseteq \bigcup_{i=0}^{k+1-t} V_i$.) This implies $\chi^{t+1}_{G_k}(u) = \chi^{t+1}_{G_k}(v)$. Together with \eqref{eq:first:column:outside}, this yields that 
\[\bigg\{C \setminus \Big(\bigcup_{i' \in [t+1]} V_{k+2-i'}\Big) \, \bigg\vert \,  C \in \pi_k^{0}, C \setminus \Big(\bigcup_{i' \in [t+1]} V_{k+2-t}\Big) \neq \emptyset\bigg\} \subseteq \pi_k^{t+1},\]
which, together with \eqref{eq:inner:columns}, concludes the proof of the claim.
 \end{claimproof}

  We call the property described in the claim \emph{path propagation from right to left}. The proof for the following statement, \emph{path propagation from left to right}, is completely analogous. Therefore, we skip it.

 \begin{claim}[resume]
 Consider a colouring $\lambda$ of $G_k$ and its induced partition $\pi_k$ of $V(G_k)$. For $t \in \mathbb{N}_0$, let $\pi_k^t$ be the partition induced by $\chi^t_{G_k}$ on input $(G_k, \lambda)$. Suppose $G_k, \lambda, \pi_k$ satisfy the following conditions.
 \begin{enumerate}
  \item There exist $\ell \in [6]$ and $I_1, \dots, I_\ell \subseteq [6]$ such that $\bigcup_{i \in [\ell]} I_i = [6]$ and $I_i \cap I_j = \emptyset$ for $1 \leq i < j \leq \ell$ and for every $i \in [\ell]$, it holds that $\{w_{1,j'} \mid j' \in I_i\} \in \pi_k$. That is, the first column is a union of colour classes with respect to $\lambda$.
  \item $\pi_k^1 = \big\{\{w_{2,j'} \mid j' \in I_i\} \, \big\vert \, i \in [\ell]\big\}  \cup \big\{C \setminus V_{2} \, \big\vert \, C \in \pi_k, C \setminus V_{2} \neq \emptyset\big\}$.
  \item For all $C,C' \subseteq V_{2}$ with $C,C' \in \pi_k^1$, the graph $G_k[C]$ is regular and $G_k[C,C']$ is biregular.
\end{enumerate}
Then for every $t \in [k]$, it holds that
\begin{align*} 
\pi_k^{t} = &\bigcup_{i' \in [t]} \big\{\{w_{i'+1,j'} \mid j' \in I_i\} \, \big\vert \, i \in [\ell]\big\} 
\\&\cup \bigg\{C \setminus \Big(\bigcup_{i' \in [t]} V_{i'+1}\Big) \, \bigg\vert \, C \in \pi_k^0, C \setminus \Big(\bigcup_{i' \in [t]} V_{i'+1}\Big) \neq \emptyset\bigg\}.
  \uenda
\end{align*}
\end{claim}

 We are now ready to analyse the run of Colour Refinement on input $G_k$.  
 Recall that $\pi_0^t$ denotes the partition induced by $\chi_{G_0}^t$ on $V(G_0) \subseteq V(G_k)$. For the following arguments, see also Figure \ref{fig:col:ref:max:iterations}. 

 In $\pi^1_k$, the vertices are distinguished according to their degrees. We can then use path propagation from right to left to deduce that 
 \begin{align*}
 \pi_k^{k+1} = \mathrlap{\pi_0^1}\phantom{\pi_0^{10}} &\cup \big\{V_i \, \big\vert \, i \in [2,k+1]\big\} \quad \text{and thus} \quad  \pi_k^{k+2} = \pi_0^2 \cup \big\{V_i \, \big\vert \, i \in [2,k+1]\big\},
\\\pi_k^{k+3} = \mathrlap{\pi_0^3}\phantom{\pi_0^{10}} &\cup \big\{V_i \, \big\vert \, i \in [2,k+1]\big\} \quad \mathrlap{\text{and also}}\phantom{\text{and thus}} \quad \pi_k^{k+4} = \pi_0^4 \cup \big\{V_i \, \big\vert \, i \in [2,k+1]\big\}.
\intertext{Now path propagation from left to right yields that} 
 \pi_k^{2k+4} = \mathrlap{\pi_0^4}\phantom{\pi_0^{10}} &\cup \big\{\{w_{i,j} \mid j \in [4]\} \, \big\vert \, i \in [2,k\mspace{-2mu}+\mspace{-2mu}1]\big\} \cup \big\{\{w_{i,j} \mid j \in \{5,6\}\} \, \big\vert \, i \in [2,k\mspace{-2mu}+\mspace{-2mu}1]\big\}\\
 \intertext{and $\pi_k^{2k+5} = \pi_0^5 \cup (\pi_k^{2k+4} \setminus \pi_0^4)$. Again using path propagation from right to left, we get that} 
 \pi_k^{3k+6} = \mathrlap{\pi_0^6}\phantom{\pi_0^{10}} &\cup \big\{\{w_{i,j} \mid j \in \{1,2\}\} \, \big\vert \, i \in [2,k+1]\big\} 
 \\&\cup \big\{\{w_{i,j} \mid j \in \{3,4\}\} \, \big\vert \, i \in [2,k+1]\big\}
 \\&\cup \big\{\{w_{i,j} \mid j \in \{5,6\}\} \, \big\vert \, i \in [2,k+1]\big\}\\
 \intertext{and $\pi_k^{3k+7} = \pi_0^7 \cup (\pi_k^{3k+6} \setminus \pi_0^6)$. Similarly, we obtain} 
 \pi_k^{4k+8} = \mathrlap{\pi_0^8}\phantom{\pi_0^{10}} &\cup \big\{\{w_{i,j} \mid j \in \{1,2\}\} \, \big\vert \, i \in [2,k+1]\big\} 
 \\&\cup \big\{\{w_{i,j} \mid j \in \{3,4\}\} \, \big\vert \, i \in [2,k+1]\big\}
 \\&\cup \big\{\{w_{i,j}\} \, \big\vert \, i \in [2,k+1], j \in \{5,6\}\big\},
 \\[1ex]\pi_k^{4k+9} = \mathrlap{\pi_0^9}\phantom{\pi_0^{10}} & \cup (\pi_k^{4k+8} \setminus \pi_0^8),
 \\[1ex]\pi_k^{5k+10} = \pi_0^{10} &\cup \big\{\{w_{i,j} \mid j \in \{1,2\}\} \, \big\vert \, i \in [2,k+1]\big\}
 \\&\cup \big\{\{w_{i,j}\} \, \big\vert \, i \in [2,k+1], j \in [3,6]\big\}, 
 \\[1ex]\pi_k^{5k+11} = \pi_0^{11} &\cup (\pi_k^{5k+10} \setminus \pi_0^{10}),
 \\[1ex]\pi_k^{5k+12} = \pi_0^{12} &\cup (\pi_k^{5k+11} \setminus \pi^{11}),
 \\[1ex]\pi_k^{6k+13} = \pi_0^{13} &\cup \big\{\{w_{i,j}\} \, \big\vert \, i \in [2,k+1], j \in [6]\big\},
 \end{align*}
 which is the discrete partition by the induction assumption for $k=0$.

 This implies that on input $G_k$, Colour Refinement takes $6k+13$ iterations to stabilise and, since $|G_k| = 6k + 14$, it holds that $\WL_1(G_k) = n-1$, where $n = |G_k|$.

 In the third and fourth family from the theorem, the three 0-pairs take on the role of $V_{k+2}$, the $(k+2)$-nd column, which initiates the first path propagation. The proofs for those families are up to index changes essentially analogous to the one just presented. 

 In the fifth and sixth family, there are more 0-pairs. We sketch the splitting. In those families, in $\pi^1$, there is one partition class formed by all vertices contained in 0-pairs. The second partition class contains all other vertices. In $\pi^2$, the vertices contained in the two adjacent 0-pairs as well as the vertices contained in $\max(\prec)$ (i.e.\ the rightmost 0) form a separate partition class, while the class consisting of all vertices not contained in 0-pairs is not split. Now, like in the other families, those three 0-pairs take up the role of $V_{k+2}$, initiate the first path propagation and the proof proceeds similarly as above.
\end{proof}

\begin{corollary}\label{cor:main:infinite}
There are infinitely many $n \in \mathbb{N}$ with $\WL_1(n) = n-1$.
\end{corollary}

\begin{corollary}\label{cor:even:max:iterations}
 For every even $n \in \mathbb{N}_{\geq 12}$ such that $n=12$ or $n \bmod 18 \notin \{6,12\}$, there is a long-refinement graph $G$ with $|G| = n$. The graph $G$ can be chosen to satisfy $\deg(G) = \{2,3\}$.
\end{corollary}

\begin{proof}
The string representation S011XX covers $n=12$. The first infinite family from Theorem \ref{thm:infinite:families} covers all even $n \in \mathbb{N}_{\geq 14}$ with $n = 2 \mod 6$, i.e.\ with $n \bmod 18 \in \{2,8,14\}$. The second and the third infinite family both cover all even $n \in \mathbb{N}_{\geq 16}$ with $n = 4 \mod 6$, i.e.\ with $n \bmod 18 \in \{4,10,16\}$. The fourth and the fifth infinite family cover all even $n \in \mathbb{N}_{\geq 18}$ with $n = 0 \mod 18$. Thus, among the even graph orders larger than 10, only the ones with $n \bmod 18 \in \{6,12\}$ remain not covered.
\end{proof}

We now turn to the long-refinement graphs $G$ of odd order with vertex degrees in $\{2,3\}$. If the graph has odd size, we cannot represent it just with pairs. For this, we relax Assumption \ref{ass:max:iterations} as follows.

\begin{assumption}\label{ass:odd:max:iterations}
$G$ is a long-refinement graph with $\deg(G) = \{2,3\}$ and such that there is an $i \in \mathbb{N}_0$ for which $\pi^i$ contains only pairs and at most one singleton.
\end{assumption}

We maintain the vocabulary and notation from the long-refinement graphs of even order, i.e.\ 0, 1, S, X will be used in the same way as before. However, in order to fully describe the odd-size graphs via strings, we have to extend the string alphabet by fresh letters $\hat 1 $ and $\mathrm{\hat X}$, which represent particular pairs with attached vertices as follows. For a string $\hat \Xi \colon [\ell] \rightarrow \mathrm{\{0,1,S,X, \hat 1, \hat X\}}$, we define the \emph{base string} $\Xi$ as the string obtained by removing hats, i.e.\ by replacing every $\hat 1$ with a $1$ and every $\mathrm{\hat X}$ with an X. Let $I(\hat \Xi) \subseteq [\ell]$ be the set of positions $i$ with $\hat \Xi(i) \in \{\hat 1, \hat X\}$. If in the base graph $G(\Xi)$, every vertex pair corresponding to a position in $I(\hat \Xi)$ (a \emph{hat vertex pair}) is adjacent, we call $\hat \Xi$ a \emph{hat string}.

 Similarly as for the even-size long-refinement graphs, to every hat string $\hat\Xi$, we assign a graph $G(\hat\Xi)$. We obtain the graph $G(\hat \Xi)$ by subdividing in $G(\Xi)$ each edge connecting a hat vertex pair with a new fresh vertex, which we call a \emph{hat}. For a hat $\hat v$, we call the neighbourhood $N(\hat v) \subseteq V\big(G(\hat \Xi)\big)$ the \emph{hat base} of $\hat v$. Note that every vertex in the hat base has degree 3 since it already has degree 3 in $G(\Xi)$ (cf.\ Notation \ref{not:string}). Also, a hat always has degree 2 and thus, with respect to $\chi^1_{G(\hat \Xi)}$, it has a different colour than its hat base. 

  Graphically, we represent a hat with a loop attached to the corresponding hat vertex pair, which we subdivide by inserting a small vertex that represents the hat (see Figure \ref{fig:sa11xx}). It is not difficult to see that every graph $G(\hat \Xi)$ corresponding to a hat string $\hat \Xi$ has exactly one hat and that the hat is the first vertex forming a singleton colour class during the execution of Colour Refinement on $G(\hat \Xi)$. Thus, $|G(\hat \Xi)| = |G(\Xi)| + 1$.

\begin{figure}
\newcommand\vDist{.40cm}
\newcommand\graphwidth{.3\textwidth}
\begin{tikzpicture}
\node[anchor=north west,text width=\textwidth-1pt,inner sep=0] (line1)
  {\includegraphics[page=7,width=\graphwidth]{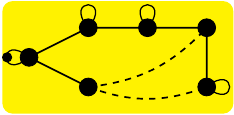} \hfill
   \includegraphics[page=8,width=\graphwidth]{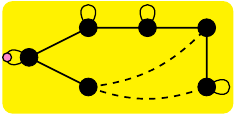} \hfill
   \includegraphics[page=9,width=\graphwidth]{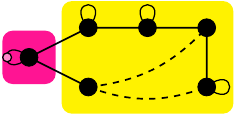}};
\draw ($(line1.south west)+(0,-\vDist/2)$) -- ($(line1.south east)+(0,-\vDist/2)$);
\node[anchor=north west,text width=\textwidth-1pt,inner sep=0] (line2) at ($(line1.south west)-(0,\vDist)$)
  {\includegraphics[page=10,width=\graphwidth]{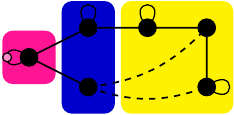} \hfill
   \includegraphics[page=11,width=\graphwidth]{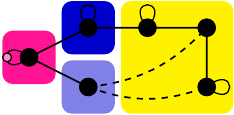} \hfill
   \includegraphics[page=12,width=\graphwidth]{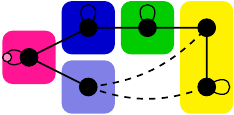}};
\draw ($(line2.south west)+(0,-\vDist/2)$) -- ($(line2.south east)+(0,-\vDist/2)$);
\node[anchor=north west,text width=\textwidth-1pt,inner sep=0] (line3) at ($(line2.south west)-(0,\vDist)$)
  {\includegraphics[page=13,width=\graphwidth]{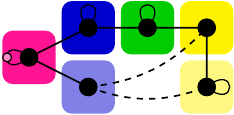}};
\end{tikzpicture}
\caption{A visualisation of the graph with string representation $\mathrm{S\hat 111XX}$ and the evolution of the colour classes in the first 6 Colour Refinement iterations on the graph.}\label{fig:sa11xx}
\end{figure}

\begin{theorem}\label{thm:odd:size:max:iterations}
For every string $\hat\Xi$ contained in the following sets, the graph $G(\hat \Xi)$ is a long-refinement graph.
\begin{itemize}
  \item $\{\mathrm{S\hat 111XX}\}$
  \item $\{\mathrm{S0X1\hat X}\} \cup \{\mathrm{S1^{\mathit{k}}1011^{\mathit{k}}X1X1^{\mathit{k}}\hat 1} \mid k \in \mathbb{N}_0\}$
  \item $\{\mathrm{S110X\hat X}\} \cup \{\mathrm{S111^{\mathit{k}}1011^{\mathit{k}}XX1^{\mathit{k}}\hat 1} \mid k \in \mathbb{N}_0\}$
  \item $\{\mathrm{S1^{\mathit{k}}01^{\mathit{k}}1XX1^{\mathit{k}}\hat 1} \mid k \in \mathbb{N}_0\}$
  \item $\{\mathrm{S(011)^{\mathit{k}}00(110)^{\mathit{k}}X\hat 1X(011)^{\mathit{k}}0} \mid k \in \mathbb{N}_0\}$
\end{itemize}
\end{theorem}

\begin{proof}[Proof sketch]
The proof techniques for the infinite families are very similar to the ones presented for the families from Theorem \ref{thm:infinite:families}. Therefore, we only sketch the proof on two concrete examples containing $\hat 1$ and $\hat X$, respectively. 
To be able to refer to vertices explicitly, recall the formal definition of $G \coloneqq G(\Xi)$ from Section \ref{sec:encodings}. Since for a hat string $\hat \Xi$, it holds that $|G(\hat \Xi)| = |G(\Xi)| + 1$, we can use the same indexing of vertices, additionally letting $\hat v$ be the unique hat in $G(\hat \Xi)$. 

In the graph $G \coloneqq G(\mathrm{S\hat 111XX})$ (cf.\ Figure \ref{fig:sa11xx}), the only vertex of degree 2 is $\hat v$. Thus, in $\pi^1$, it forms a singleton colour class. In $\pi^2$, the hat base $\{v_{2,1}, v_{2,2}\}$ forms a new colour class. 

In general, for $i \in \mathbb{N}$, we have that $\pi^i_G = \pi^{i-1}_{G'} \cup \{\hat v\}$, where $G' = G\mathrm{(S011XX)}$ (cf.\ Theorem \ref{thm:infinite:families} and Figure \ref{fig:s011xx}).

Thus,
 \begin{align*}
 \pi^1 &= \big\{\{\hat v\}, V(G) \setminus \{\hat v\}\big\},\\
 \pi^2 &= \big\{\{\hat v\},\{v_{2,1}, v_{2,2}\}, V(G) \setminus \{\hat v, v_{2,1}, v_{2,2}\}\big\},
 \\\pi^3 &= \big\{\{\hat v\},\{v_{2,1}, v_{2,2}\}, \{v_{i,j} \mid i \in \{1,3\}, j \in [2]\}, \{v_{i,j} \mid i \in [4,6], j \in [2]\}\big\},\\
 \pi^4 &= \big\{\{\hat v\},\{v_{1,1}, v_{1,2}\}, \{v_{2,1}, v_{2,2}\}, \{v_{3,1}, v_{3,2}\}, \{v_{i,j} \mid i \in [4,6], j \in [2]\}\big\},\\ 
 \pi^5 &= \big\{\{\hat v\},\{v_{1,1}, v_{1,2}\}, \{v_{2,1}, v_{2,2}\}, \{v_{3,1}, v_{3,2}\}, \{v_{4,1}, v_{4,2}\}, \{v_{i,j} \mid i \in [5,6], j \in [2]\}\big\},\\
  \pi^6 &= \big\{\{\hat v\}\big\} \cup \big\{\{v_{i,j} \mid j \in [2]\} \mid i \in [6]\big\}.
  \end{align*}

Now in 6 further iterations, the splitting of the pairs is propagated linearly according to the order $\prec$.

Next, let $G$ be the graph $G(\mathrm{S0X1\hat X})$, i.e.\ a member of the first infinite family. It has three vertices of degree 2, namely $\hat v$, $v_{2,1}$, and $v_{2,2}$, which therefore form a colour class in $\pi^1$. Also, the vertices contained in the 1-pair are the only vertices of degree 3 that are not adjacent to any vertex of degree 2. Thus, 
 \begin{align*}
 \pi^2 &= \big\{\{\hat v,v_{2,1},v_{2,2}\}, \{v_{4,1},v_{4,2}\}, V(G) \setminus \{\hat v,v_{2,1},v_{2,2},v_{4,1},v_{4,2}\}\big\},\\
\intertext{and similarly,}
 \pi^3 &= \big\{\{\hat v,v_{2,1},v_{2,2}\}, \{v_{4,1},v_{4,2}\}, \{v_{1,1},v_{1,2}\}, \{v_{i,j} \mid i \in \{3,5\}, j \in [2]\}\big\}.\\
\intertext{Now the hat forms a singleton since, in contrast to the vertices of the 0-pair, it is not adjacent to any vertex in the S-pair. We obtain:}
 \pi^4 &= \big\{\{\hat v\}, \{v_{2,1},v_{2,2}\}, \{v_{4,1},v_{4,2}\}, \{v_{1,1},v_{1,2}\}, \{v_{i,j} \mid i \in \{3,5\}, j \in [2]\}\big\},\\
 \pi^5 &= \big\{\{\hat v\}\big\} \cup \big\{\{v_{i,j} \mid j \in [2]\} \mid i \in [5]\big\}.
  \end{align*}
Then in 5 further iterations, the splitting of the pairs is propagated linearly according to the order $\prec$.
\end{proof}

\begin{corollary}\label{cor:odd:max:iterations}
 For every odd $n \in \mathbb{N}_{\geq 11}$ with $n \bmod 18 \notin \{3,9\}$, there is a long-refinement graph $G$ with $|G| = n$. The graph $G$ can be chosen to satisfy $\deg(G) = \{2,3\}$.
\end{corollary}

\begin{proof}
The string representation $\mathrm{S\hat 111XX}$ covers $n=13$. The first infinite family covers all odd $n \in \mathbb{N}_{\geq 11}$ with $n = 5 \mod 6$, i.e.\ with $n \bmod 18 \in \{5,11,17\}$. The second and the third infinite family both cover all all odd $n \in \mathbb{N}_{\geq 13}$ with $n = 1 \mod 6$, i.e.\ with $n \bmod 18 \in \{1,7,13\}$. The fourth infinite family covers all odd $n \in \mathbb{N}_{\geq 15}$ with $n = 15 \mod 18$. Thus, among the odd orders larger than 10, only the ones with $n \bmod 18 \in \{3,9\}$ family are skipped.
\end{proof}

We summarise the results from Corollaries \ref{cor:even:max:iterations} and \ref{cor:odd:max:iterations}.

 \begin{corollary}\label{cor:max:iterations:23}
 For every $n \in \mathbb{N}_{\geq 11}$ such that $n = 12$ or $n \bmod 18 \notin \{3,6,9,12\}$, there is a long-refinement graph $G$ with $|G| = n$. The graph $G$ can be chosen to satisfy $\deg(G) = \{2,3\}$.
 \end{corollary}

 The following lemma allows to cover more graph sizes.

\begin{lemma}\label{lem:extension:max:iteration:number}
Let $n \in \mathbb{N}$ be arbitrary. Suppose there is a long-refinement graph $G$ such that $|G| = n$. If there is a $d \in \mathbb{N}$ with $\deg(G) = \{d,d+1\}$ such that $|\{v \in V(G) \, \vert \, \deg(v) = d\}| \neq d+1$, then there is also a long-refinement graph $G'$ with $|G'| = n+1$.
\end{lemma}

\begin{proof}
We can insert an isolated vertex $w$ into $G$ and insert edges from $w$ to every vertex $v \in V(G)$ with $\deg(v) = d$. In the new graph $G'$, the vertex $w$ has a degree other than $d+1$, whereas all other vertices have degree $d+1$. Thus, the colour classes in $\pi^1_{G'}$ are $\{w\}$ and $V(G') \setminus \{w\}$. After the second iteration, the neighbours of $w$ are distinguished from all other vertices, just like they are in $\pi^1_G$. Inductively, it is easy to see that for $i \in \mathbb{N}$, it holds that $\pi^i_{G'} = \pi^{i-1}_{G} \cup \big\{\{w\}\big\}$. Thus, Colour Refinement takes $n-1+1 = n$ iterations to stabilise on $G'$.
\end{proof}

\begin{corollary}\label{cor:all:odd:max:iterations}
 For every odd $n \in \mathbb{N}_{\geq 11}$, there is a long-refinement graph $G$ with $|G| = n$. 
\end{corollary}

\begin{proof}
By Corollary \ref{cor:odd:max:iterations}, it suffices to provide long-refinement graphs of order $n$ for every odd $n \in \mathbb{N}_{\geq 11}$ with $n \bmod 18 \in \{3,9\}$. We will accomplish this by applying Lemma \ref{lem:extension:max:iteration:number} to suitable graphs of orders $n'$ with $n' \bmod 18 \in \{2,8\}$. Every graph $G$ with a string representation contained in one of the infinite families from Theorem \ref{thm:infinite:families} has an even number of vertices of degree 2. In particular, it satisfies $|\{v \in V(G) \mid \deg(v) = 2\}| \neq 3$. Furthermore, every even graph order larger than 10 not covered by Corollary \ref{cor:even:max:iterations} is a multiple of 6. Hence, since $N \coloneqq \{n \in \mathbb{N}_{\geq 18} \mid n \bmod 18 \in \{2,8\}\}$ contains only even numbers and no multiples of 6, for every $n' \in N$, there is a graph of order $n'$ that satisfies the prerequisites of Lemma \ref{lem:extension:max:iteration:number} with $d = 2$. (Actually, we can cover all of these graph orders with the family $\{\mathrm{S1^{\mathit{k}}001^{\mathit{k}}X1X1^{\mathit{k}}0} \mid k \in \mathbb{N}_0\}$.) 

Thus, applying the lemma, we can construct for every $n \in \mathbb{N}_{\geq 18}$ with $n \bmod 18 \in \{3,9\}$ a long-refinement graph $G'$.  
\end{proof}

Note that, since we apply Lemma \ref{lem:extension:max:iteration:number} to close the gaps, we cannot guarantee anymore that the vertex degrees are 2 and 3, as we could in Corollary \ref{cor:max:iterations:23}.

We are ready to prove Theorem \ref{thm:main:infinite}.

\begin{proof}[Proof of Theorem \ref{thm:main:infinite}]
The theorem follows from combining Corollaries \ref{cor:even:max:iterations} and \ref{cor:all:odd:max:iterations} with Theorem \ref{thm:goedicke}.
\end{proof}

Although the corollary leaves some gaps, we can deduce a new lower bound on the number of Colour Refinement iterations until stabilisation, which is optimal up to an additive constant of 1.

\begin{proof}[Proof of Theorem \ref{thm:main:n-2}]
By Theorem \ref{thm:main:infinite}, we only need to consider the numbers $n \geq 24$ for which $n \bmod 18 \in \{6,12\}$. Since these numbers are all even, for every such $n$, we can use Corollary \ref{cor:all:odd:max:iterations} to obtain a long-refinement graph $G'$ with $|G'| = n-1$. Let $v$ be a fresh vertex not contained in $V(G')$. Now define $G \coloneqq (V,E)$ with $V \coloneqq V(G') \cup \{v\}$ and $E \coloneqq E(G')$. Then $\pi^i_G = \pi^i_{G'} \cup \{\{v\}\}$ for every $i \in \mathbb{N}$. In particular, $\WL_1(G) = \WL_1(G') = n-2$.
\end{proof}

\section{Conclusion}
With Theorem \ref{thm:main:n-2}, it holds for all $n \in \mathbb{N}_{\geq 10}$ that $\WL_1(n) \geq n-2$. In particular, this proves that the trivial upper bound $\WL_1(n) = n-1$ is tight, up to an additive constant of 1.

For infinitely many graph sizes, the graph $G$ can even be chosen to have vertex degrees 2 and 3, as Theorems \ref{thm:infinite:families} and \ref{thm:odd:size:max:iterations} show. 
We applied Lemma \ref{lem:extension:max:iteration:number} to cover some of the remaining sizes. However, no order $n \in \mathbb{N}_{\geq 18}$ with $n \bmod 18 \in \{6,12\}$ is covered by Theorem \ref{thm:infinite:families}. Also, for $|G| \in \{n \in \mathbb{N}_{\geq 18} \mid n \bmod 18 \in \{5,11\}\}$, all the long-refinement graphs $G$ we have found satisfy $|\{v \in V(G) \, \vert \, \deg(v) = 2\}| = 3$ (see the first infinite family in Theorem \ref{thm:odd:size:max:iterations}). Thus, these graphs do not satisfy the prerequisites of Lemma \ref{lem:extension:max:iteration:number}. Note that we cannot apply the construction from the proof if $|\{v \in V(G) \, \vert \, \deg(v) = d\}| = d+1$, since the new graph $G'$ would be $(d+1)$-regular and would thus satisfy $\WL_1(G') = 0$. Hence, it is not clear how to apply our techniques to construct a long-refinement graph of order 24. In fact, our computational results yield that there are no long-refinement graphs with 24 vertices and maximum degree 3. Altogether, the values $n \in \mathbb{N}_{\geq 24}$ with $n \bmod 18 \in\{6,12\}$ are precisely the graph orders for which it remains open whether there is a graph $G$ with $\WL_1(G) = |G|-1$. 

A related question is for which values $d_1 \neq d_2$ there are long-refinement graphs $G$ with $\deg(G) = \{d_1,d_2\}$. It would be nice to know whether we have actually found all long-refinement graphs $G$ with $\deg(G) = \{2,3\}$. Similarly, we can ask for long-refinement graphs when fixing other parameters. For example, since all long-refinement graphs that we found have girth at most 4, it would be interesting to know whether there exists any long-refinement graph with larger girth, or infinite families with unbounded girth.

Also, in the light of the graph isomorphism problem, it is a natural follow-up task to find for each order $n$ pairs of non-isomorphic graphs $G$, $H$ for which it takes $n-1$ Colour Refinement iterations to distinguish the graphs from each other. A first step towards this goal is the search for pairs of long-refinement graphs of equal order. It is easy to see that for infinitely many $n$, Theorems \ref{thm:infinite:families} and \ref{thm:odd:size:max:iterations} yield such pairs of graphs. Still, for example, when evaluating the colourings computed by Colour Refinement on the graphs with string representations S1100XX0 and S001XX10, they differ after less than $n-1$ iterations. To see this, observe that in $G(\mathrm{S1100XX0})$, all vertices of degree 2 have paths of length 3 to a vertex of degree 2 whose inner vertices only have degrees other than 2. This is not the case for $G(\mathrm{S001XX10})$ and this property is detected by Colour Refinement after at most 4 iterations.) Thus, finding for $n \in \mathbb{N}_{\geq 10}$ two graphs of order $n$ which Colour Refinement only distinguishes after $n-1$ iterations remains a challenge.

\bibliography{long_refinement}

\end{document}